\documentclass[11pt]{elsarticle}
\usepackage{amsmath,amssymb,amsthm}
\usepackage{color}
\usepackage[hmargin=1in]{geometry}
\usepackage{graphicx,subfigure}
\usepackage{natbib}
\usepackage{setspace}
\newcommand{\s}{\Xi_{\mu,\sigma}^r}

\newcommand{\PP}{{\mathord{I\kern -.33em P}}}
\newcommand{\EE}{{\mathord{I\kern -.33em E}}}
\newcommand{\RR}{{\mathord{I\kern -.33em R}}}
\def\P{\mathbb{P}} 
\def\E{\EE} 
\def\Q{{\mathbb Q}} 
\def\1{{\mathbf 1}} 
\def\F{{\mathcal F}} 

\def\tdk{{\tau_D(k)}}
\def\tuk{{\tau_U(k)}}

\newtheorem{theorem}{Theorem}[section]

\newtheorem{lemma}[theorem]{Lemma}

\newtheorem{proposition}[theorem]{Proposition}
\theoremstyle{remark}
\newtheorem{remark}[theorem]{Remark}

\journal{}
\begin{document}
\begin{frontmatter}
\title{Stochastic Modeling and Fair Valuation of Drawdown Insurance}

\author[gcma]{Hongzhong~Zhang\corref{cor1}}
\ead{hzhang@stat.columbia.edu}
\cortext[cor1]{Corresponding author}

\author[ieor]{Tim~Leung}
\ead{leung@ieor.columbia.edu}

\author[bk]{Olympia~Hadjiliadis}
\ead{ohadjiliadis@brooklyn.cuny.edu}

\address[gcma]{Department of Statistics, Columbia University, 1255 Amsterdam Avenue, New York, NY 10027}
\address[ieor]{Department of Industrial Engineering \& Operations Research, Columbia University, New York, NY 10027}
\address[bk]{Department of Mathematics, Brooklyn College and the Graduate Center, C.U.N.Y. Brooklyn, NY 11209}

\begin{abstract} This paper studies the stochastic modeling of  market drawdown events and the fair valuation of insurance contracts based on drawdowns. We model the asset drawdown process as the current relative distance from the historical maximum of the asset value.  We first consider a vanilla insurance contract whereby  the protection buyer pays a constant premium over time to  insure against a drawdown  of a  pre-specified level. This leads to the analysis of the conditional Laplace transform of the drawdown time, which will serve as the building block for   drawdown insurance with early cancellation or drawup contingency. For the cancellable drawdown insurance, we  derive  the investor's optimal cancellation timing in terms of a two-sided first passage time of the underlying drawdown process. Our   model can also be applied to insure against a drawdown by a defaultable stock. We provide analytic formulas for the fair premium and  illustrate the impact of default risk.
\end{abstract}

\begin{keyword}
 Drawdown insurance; Early cancellation; Optimal stopping; Default risk. \\
 JEL subject classification: C61, G01, G13, G22.
  \end{keyword}

\end{frontmatter}

\section{Introduction}The recent financial crisis has been marked with series of  sharp falls in asset  prices triggered by, for example, the S\&P  downgrade of US debt, and default speculations of European countries.  Many individual and institutional investors are wary of large market drawdowns  as  they not only    lead to portfolio losses and liquidity shocks, but also   indicate potential  imminent  recessions.  As is well known,  hedge fund managers  are typically compensated based on the fund's outperformance over the last record maximum, or the high-water mark (see \cite{Agarwal09,Goetzmann,GrosZhou,Sorn}, among others).  As such, drawdown events can directly affect the manager's income. Also, a major drawdown may also trigger a surge in fund redemption by investors, and lead to the manger's job termination. Hence, fund managers have strong incentive to seek insurance against drawdowns.

These market phenomena have motivated the application of drawdowns as path-dependent risk measures, as discussed in \cite{MagdAtiy}, \cite{PospVece}, among others. On the other hand,  Vecer \cite{Vece06, Vece07} argues that some market-traded contracts, such as vanilla and lookback puts, ``have only limited ability to insure the market drawdowns." He studies through simulation the returns of   calls and puts written on the underlying asset's maximum drawdown, and discusses dynamic trading strategies to hedge against a drawdown associated with a single   asset or index.   The recent work  \cite{CarrZhanHaji} provides non-trivial  static strategies using market-traded barrier digital options  to  approximately synthesize  a European-style digital option  on a drawdown event. These observations suggest that drawdown protection  can be useful for both institutional and individual investors, and there is an interest  in synthesizing  drawdown  insurance.

%

In the current paper,  we discuss the  stochastic modeling of drawdowns and  study the  valuation  of a number of insurance contracts against drawdown events. More precisely, the drawdown process  is defined as the current relative   drop of an asset value  from its historical maximum. In its simplest form, the drawdown insurance involves a continuous  premium payment by the investor (protection buyer) to insure a drawdown of an underlying asset value to  a pre-specified level.

 In order to provide the investor with more flexibility in managing  the path-dependent drawdown risk,  we incorporate the right to terminate the contract early. This  early cancellation feature  is similar to the surrender right  that arises in many common insurance products such as equity-indexed annuities (see e.g. \cite{Cheung2005}, \cite{Moore2009}, \cite{MooreYoung1}). Due to the timing flexibility, the investor may stop the premium payment if he/she finds that a drawdown is unlikely to occur (e.g. when the underlying price continues to rise).  In our analysis, we rigorously show that the investor's optimal cancellation timing is based on a non-trivial first passage time of the underlying drawdown process. In other words, the investor's  cancellation strategy and valuation of the contract will depend not only on current value of the underlying asset, but also its distance from the historical maximum. Applying the theory of optimal stopping as well as analytical properties of drawdown processes, we  derive the optimal cancellation threshold  and illustrate it through numerical examples.

Moreover, we consider a related insurance contract that protects the investor from a  drawdown  preceding a drawup.  In other words, the insurance contract  expires early if a drawup event occurs prior to a drawdown. From the investor's perspective, when  a drawup is realized, there is little  need to insure against a drawdown. Therefore, this drawup contingency automatically stops the premium payment and is an attractive feature that  will potentially reduce the  cost of drawdown insurance.

 Our model can also readily extended to incorporate the default risk associated with the underlying asset. To this end, we observe that a  drawdown can be triggered by  a continuous price movement as well as a jump-to-default event. Among other results, we provide  the formulas for the fair premium of the drawdown insurance, and analyze the impact of default risk on the valuation of drawdown insurance.



In existing literature, drawdowns also arise in a number of financial applications. Pospisil and Vecer \cite{PospVece} apply PDE methods to investigate the sensitivities of portfolio values and hedging strategies with respect to drawdowns and drawups.   Drawdown processes have also been incorporated into trading constraints for portfolio optimization (see e.g. \cite{GrosZhou,CvitKara,ChekUryaZaba}). Meilijson \cite{Meil} discusses  the role of  drawdown  in  the  exercise time for a certain look-back American put option.  Several     studies focus on some related concepts of drawdowns, such as maximum drawdowns \cite{DSY, MagdAtiy,Vece06,Vece07}, and speed of market crash \cite{ZhanHadjSMC}.
 On the other hand, the statistical modeling of drawdowns and drawups is also of practical importance, and we refer to the recent studies \cite{Leal,Johansen,RebonatoGaspari}, among others.

 For our valuation problems, we often work with  the joint law of drawdowns and drawups. To this end,  some related formulas from \cite{HadjVece}, \cite{PospVeceHadj}, \cite{ZhanHadj},  and \cite{Zhan13} are useful.   Compared to the existing literature and our prior work, the current paper's contributions are  threefold. First, we derive  the fair premium   for insuring  a number of drawdown events,  with both finite and infinite maturities, as well as new provisions like  drawup contingency and early termination.  In particular, the early termination option leads to the analysis of a new optimal stopping problem (see Section \ref{sect-cancellableinsurance}). We rigorously solve for the optimal termination  strategy, which can be expressed in terms of first passage time of a drawdown process. Furthermore, we incorporate the underlying's default risk -- a feature absent in other related studies on drawdown --  into our analysis,  and  study its impact on the drawdown insurance premium.

 The paper is structured as follows.  In Section   \ref{sect-formulation}, we describe a stochastic model for drawdowns and drawups, and formulate the valuation of a vanilla drawdown insurance. In Sections \ref{sect-cancellableinsurance} and \ref{sect-drawuptoo},  we study, respectively,  the cancellable  drawdown insurance and drawdown insurance with drawup contingency. As extension, we discuss the valuation of  drawdown insurance  on a defaultable stock in Section \ref{sect-def}.  Section \ref{sect-conclude} concludes the paper. We include the proofs for a number of lemmas in Section \ref{sect_proofs}.

\section{Model for Drawdown Insurance}\label{sect-formulation}

We fix a complete filtered probability space $(\Omega, \F, (\F_t)_{t\ge 0}, \Q)$ satisfying the usual conditions. The risk-neutral pricing measure $\Q$ is used for our valuation problems.  Under the measure $\Q$, we model a risky asset $S$ by  the geometric Brownian motion \begin{eqnarray}\label{GBM} \frac{d S_t}{S_t} = r dt + \sigma dW_t\end{eqnarray}
where $W$ is  a standard Brownian motion under $\Q$ that generates the filtration  $(\F_t)_{t\ge 0}$.

Let us denote $\overline{S}$ and $\underline{S}$, respectively,  to be the processes for the running maximum and running minimum of $S$. When writing the contract, the  insurer may use the historical  maximum $\overline{s}$  and  minimum $\underline{s}$ recorded from a prior reference period.  Consequently, at the time of contract inception, the reference  maximum $\overline{s}$, the reference minimum $\underline{s}$ and the stock price need not coincide. This is illustrated in Figure \ref{fig:SPX}.

The running maximum and running minimum processes associated with $S$ follow\footnote{Herein, we denote $a\vee b=\max(a,b)$ and $a\wedge b=\min(a,b)$.},
\begin{eqnarray}\overline{S}_t=\overline{s}\vee\Big(\sup_{s\in[0,t]}S_s\Big),\quad \underline{S}_t=\underline{s}\wedge\Big(\inf_{s\in[0,t]}S_s\Big). \end{eqnarray}
We define the stopping times
\begin{equation}
 \varrho_D(K)=\inf\{t \geq 0 : \overline{S}_t/{S_t}    \geq K \}\quad \text{ and }  
 \quad \varrho_U(K)=\inf\{t \geq 0 : {S_t}/{\underline{S}_t}\ge K\}\label{tauDkr},
\end{equation}
respectively as the first times that $S$ attains a \emph{relative drawdown}  of $K$ units and a \emph{relative drawup} of $K$ units. Without loss of generality, we assume that $1\le\overline{s}/\underline{s}<K$ so that $\varrho_D(K)\wedge\varrho_D(K)>0$, almost surely.

To facilitate our analysis, we shall work with log-prices. Therefore, we define $X_t = \log S_t$ so that
\begin{eqnarray}\label{DBM}X_t=x + \mu t+\sigma W_t,\end{eqnarray} where $x = \log S_0$ and  $\mu =  r - \frac{\sigma^2}{2}$.
Denote by $\overline{X}_t=\log\overline{S}_t$ and $\underline{X}_t=\log\underline{S}_t$ to be, respectively, the  running maximum and running minimum of the log price process.  Then, the  relative drawdown and drawup of $S$ are equivalent to  the absolute drawdown and drawup of the log-price $X$, namely,
\begin{equation}
  \tau_D(k)=\inf\{t \geq 0 : D_t\geq k \}\quad\text{ and }
\quad\tau_U(k)=\inf\{t \geq 0 : U_t\ge k\}, \label{tauDk}
\end{equation} where $k = \log K$  (see \eqref{tauDkr}), 
  $D_t=\overline{X}_t-X_t$  and   $U_t=X_t-\underline{X}_t.
$
Note that under the current model the stopping times $\tau_D(k)=\varrho_D(K)$ and $\tau_U(k)=\varrho_U(K)$, and they do not depend on $x$ or equivalently the initial stock price.

 \begin{figure}
  \begin{center}
\includegraphics[width=0.7\textwidth]{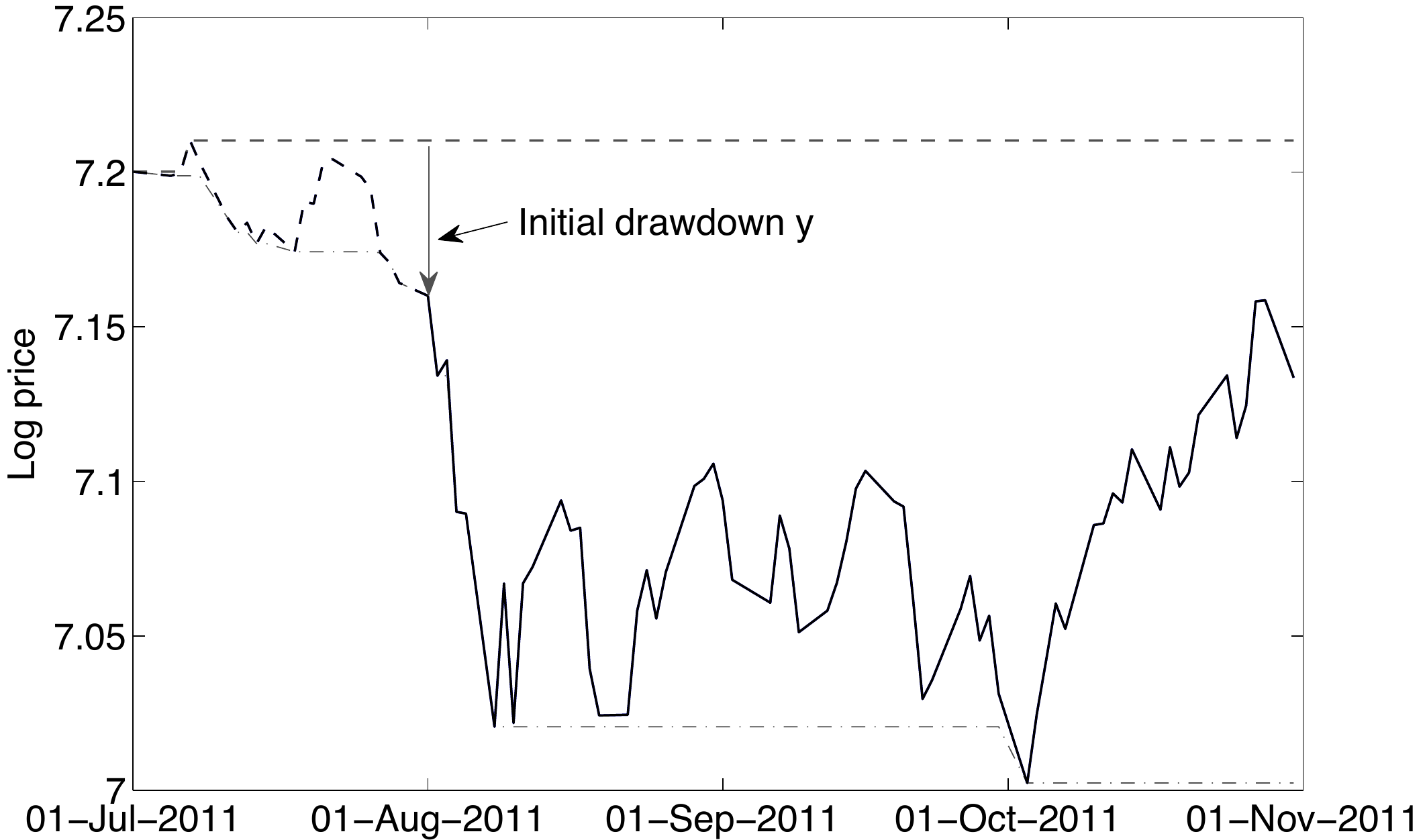}
\end{center}
\begin{small}\caption{Daily log-price of S\&P Index from 07/01/2011 to 11/01/2011. For illustration, July is used as the reference period to record the historical running maximum and minimum. At the end of the reference period, the running maximum  $\overline{s}=7.21$ and the log-price $x=7.16$, so the initial drawdown $y=0.05$.  We remark that the large drawdown  in August 2011 due to the  downgrade of US debt by S\&P.}\label{fig:SPX} \end{small}
\end{figure}

\subsection{Drawdown Insurance and Fair Premium}
We now consider an insurance contract based on a drawdown event. Specifically,  the protection buyer who seeks insurance on a drawdown event of size $k$ will pay a constant premium payment $p$ continuously over time until  the drawdown time $\tau_D(k)$. In return,  the protection buyer will receive the insured amount $\alpha$ at time $\tau_D(k)$. Here, the values $p$, $k$ and $\alpha$ are pre-specified at the contract inception. The contract value of this   drawdown insurance is
\begin{align}f(y;\,p) &=    \E\left\{- \int_0^{\tdk }e^{-rt}p\, dt + \alpha e^{-r\tdk}\,|\,D_0=y\right\}\label{F1}\\
&=\frac{p}{r} -\bigg(\alpha+\frac{p}{r}\bigg)\xi(y), \end{align}
where $\xi$ is the  conditional Laplace transform of $\tau_D(k)$ defined by
\begin{align}\label{def_xi}\xi(y):=\E\{e^{-r\tau_D(k)}\, |\, D_0= y\},\quad 0\le y\le k. \end{align}
This amounts to computing the conditional Laplace transform $\xi$, which admits a closed-form formula as we show next.
 \begin{proposition} The  conditional Laplace transform function $\xi(\cdot)$ is given by \vspace*{-1ex}
  \begin{align}\label{zdprop}
  \xi(y)
  &=e^{\frac{\mu}{\sigma^2}(y-k)}\frac{\sinh(\Xi_{\mu,\sigma}^ry)}{\sinh(\Xi_{\mu,\sigma}^{r}k)}+e^{\frac{\mu}{\sigma^2}y}\frac{\sinh(\Xi_{\mu,\sigma}^r(k-y))}{\sinh(\Xi_{\mu,\sigma}^rk)}\frac{e^{-\frac{\mu}{\sigma^2}k}\Xi_{\mu,\sigma}^r}{\Xi_{\mu,\sigma}^r\cosh(\Xi_{\mu,\sigma}^rk)-\frac{\mu}{\sigma^2}\sinh(\Xi_{\mu,\sigma}^rk)},~ 0\le y\le k.
\end{align}

where
$\Xi_{\mu,\sigma}^r=\sqrt{\frac{2r}{\sigma^2}+\frac{\mu^2}{\sigma^4}}.$
  \end{proposition}
  \begin{proof}Define the first time that the drawdown process $(D_t)_{t\ge 0}$ decreases to a level $\theta \ge 0$ by
  \begin{align}\label{tauDtheta}\tau_D^-(\theta):=\inf\{t\ge 0\,:\,D_t\le \theta\}.\end{align}
By the strong Markov property of process $D$ at  $\tau_D^-(0)$,  we have that for $t\le \tau_D(k)$,
\begin{align}\xi(D_t)&=\E\{e^{-r\tau_D(k)}\,|\,{D}_t\}\notag\\
  &=\E\{e^{-r\tdk}\1_{\{\tdk<\tau_D^{-}(0)\}}\,|\,D_t\}+\E\{e^{-r\tau_D^{-}(0)}
  \1_{\{\tdk>\tau_D^{-}(0)\}}\,|\,D_t\}\xi(0)\notag\\
  &=e^{\frac{\mu}{\sigma^2}(D_t-k)}\frac{\sinh(\Xi_{\mu,\sigma}^rD_t)}{\sinh(\Xi_{\mu,\sigma}^{r}k)}+e^{\frac{\mu}{\sigma^2}D_t}\frac{\sinh(\Xi_{\mu,\sigma}^r(k-D_t))}{\sinh(\Xi_{\mu,\sigma}^rk)}\xi(0).\label{zd}\end{align}
Therefore, the problem is reduced to finding $\xi(0)$, which is known (see \cite{Lehoczky77}):
\[ \xi(0)=\frac{e^{-\frac{\mu}{\sigma^2}k}\Xi_{\mu,\sigma}^r}{\Xi_{\mu,\sigma}^r\cosh(\Xi_{\mu,\sigma}^rk)-\frac{\mu}{\sigma^2}\sinh(\Xi_{\mu,\sigma}^rk)}.\]
Substituting this to \eqref{zd} yields \eqref{zdprop}.  \end{proof}

Therefore, the contract value $f(y;p)$ in \eqref{F1} is explicit given for any premium rate $p$. The fair premium $P^*$ is found from the equation $f(y;P^*)=0$, which yields
\begin{align} P^* = \frac{r \alpha \xi(y)}{1- \xi(y)}.   \label{premium1}\end{align}

\begin{remark}\label{rem_1}
 Our formulation can be adapted to the case when the drawdown insurance is paid for upfront. Indeed, we can set $p=0$ in \eqref{F1},  then the price of this contract  at time zero is $f(y;0)$. On the other hand, if the insurance premium is paid over a pre-specified period of time $T'$, rather than up to the random drawdown time, then the present value of the   premium cash flow  $\frac{p}{r} (e^{-rT'}-1)$ will replace the first term in the expectation of (2.7). In this case, setting the contract value zero at inception, the fair premium is given by $P^*(T'):= \frac{f(y;\,0) r}{1- e^{-rT'}}>0$. In Section \ref{sect-drawuptoo}, we discuss the case where the holder will stop premium payment  if a drawup event occurs prior to drawdown or maturity.
\end{remark}

For both the insurer and protection buyer, it is useful to know how long the drawdown is expected to occur. This leads us to compute the expected time to a drawdown of size $k\ge 0$, under the physical measure $\P$.  The measure $\P$ is equivalent to $\Q$, whereby the drift  of $S$ is the annualized growth rate $\nu$,  not the risk-free rate $r$. Under measure  $\P$, the  log price is \[X_t = x + \tilde{\mu} t + \sigma W^\P_t, \quad \text{ with } \tilde{\mu} = \nu - \sigma^2/2,\] where $W^\P$ is a $\P$-Brownian motion.

\begin{proposition}\label{exptime}The expected time to drawdown of size $k$ is given by
\begin{align}\label{Etdk}\E_{\P}\{\tau_D(k)|D_0=y \} = \frac{y\cdot\rho_{\tau}(y;\,k)+(y-k)\cdot e^{\frac{2\tilde{\mu}}{\sigma^2}(y-k)}\rho_{\tau}(k-y;\,k)}{\tilde{\mu}}+\rho_{\tau}(y;\,k)\cdot\frac{e^{\frac{2\tilde{\mu}}{\sigma^2}k}-\frac{2\tilde{\mu}}{\sigma^2}k-1}{(\frac{2\tilde{\mu}}{\sigma^2})^2},\end{align}
where
$\rho_{\tau}(y;\,k)\mathop{:=}e^{\frac{\tilde{\mu}}{\sigma^2}y}\frac{\sinh(\frac{\tilde{\mu}}{\sigma^2}(k-y))}{\sinh(\frac{\tilde{\mu}}{\sigma^2}k)}.$
\end{proposition}
\begin{proof}
  By  the Markov property of the process $(X_t)_{t\ge0}$, we know that
  \[\tdk=\tau_{x+y-k}\wedge\tau_{x+y}+(\tdk\circ\theta_{\tau_{x+y}})\cdot\1_{\{\tau_{x+y}<\tau_{x+y-k}\}},~\P\text{-a.s.}\]
  where $\tau_w=\inf\{t\ge 0\,:\,X_t=w\},$ and $\theta_\cdot$ is the standard Markov shift operator. If $\tilde{\mu}\neq0$, applying the optional sampling theorem to uniformly integrable martingale $(M_{t\wedge\tau_{x+y-k}\wedge\tau_{x+y}})_{t\ge0}$ with $M_t=X_t-\tilde{\mu}t$, we obtain that
  \[\E_{\P}\{\tau_{x+y-k}\wedge\tau_{x+y}|X_0=x\}=\frac{y\cdot\P(\tau_{x+y}<\tau_{x+y-k}|X_0=x)+(y-k)\cdot\P(\tau_{x+y-k}<\tau_{x+y}|X_0=x)}{\tilde{\mu}}.\]
  Moreover, using the fact that $\P(\tau_{x+y}<\tau_{x+y-k}|X_0=x)=\rho_{\tau}(y;\,k)$, $\P(\tau_{x+y-k}<\tau_{x+y}|X_0=x)=e^{\frac{2\tilde{\mu}}{\sigma^2}(y-k)}\rho_{\tau}(k-y;\,k)$ and Eq. (11) of \cite{HadjVece}:
  \[\E_{\P}\{\tdk|D_0=0\}=\frac{e^{\frac{2\tilde{\mu}}{\sigma^2}k}-\frac{2\tilde{\mu}}{\sigma^2}k-1}{(\frac{2\tilde{\mu}}{\sigma^2})^2},\]
  we conclude the proof for $\tilde{\mu}\neq0$. The case of $\tilde{\mu}=0$ is obtained by taking the limit $\tilde{\mu}\to 0$.
  \end{proof}

\section{Cancellable Drawdown Insurance}\label{sect-cancellableinsurance} As is common in insurance and derivatives markets, investors may demand  the option to voluntarily terminate their contracts early. Typical examples include American options and equity-indexed annuities with surrender rights.   In this section, we  incorporate a cancelable feature into our drawdown insurance, and investigate the  optimal timing to terminate the contract.

With  a cancellable drawdown insurance,  the  protection buyer can terminate the position by paying a constant fee $c$ anytime prior to a pre-specified drawdown of size $k$. For a notional amount of $\alpha$ with premium rate $p$, the fair valuation  of this contract  is found from the  optimal stopping problem:
\begin{align}\label{cancellableV}V(y;\,p) &= \sup_{0\leq \tau<\infty}\E\left\{- \int_0^{\tdk \wedge \tau}e^{-rt}p\, dt - c e^{-r\tau} \1_{\{\tau< \tdk\}} + \alpha e^{-r\tdk} \1_{\{\tdk\le \tau\}}\,|\, D_0 = y\right\}
\end{align} for  $y\in[0,k)$.   The  fair premium $P^*$ makes the contract value zero at inception, i.e. $V(y; P^*) =0$.

 We observe that it is never optimal to cancel and pay the fee $c$ at $\tau=\tau_D(k)$ since the contract expires and pays at $\tau_D(k)$. Hence, it is sufficient to consider a smaller set of stopping times $\mathcal{S}:=\{\tau \in \mathbb{F}\,:\, 0<\tau<\tau_D(k)\}$, which consists of  $\mathbb{F}$-stopping times strictly bounded by $\tau_D(k)$. We will show in Section \ref{sect-optcancel} that  the  set of \emph{candidate} stopping times are in fact the drawdown stopping times $\tau=\tau_D^{-}(\theta)$ indexed by their respective thresholds $\theta \in (0,k)$ (see \eqref{tauDtheta}).

\subsection{Contract Value Decomposition}\label{sect-decomp}

Next, we show that the  cancellable drawdown insurance can be decomposed into an ordinary drawdown insurance and  an American-style claim on  the drawdown insurance. This provides a key insight for the explicit computation of the contract value as well as the optimal termination strategy.
\begin{proposition}The cancellable drawdown insurance value admits the decomposition:
  \begin{align}\label{decomp}V(y;\,p)&=-f(y;\,p)+\sup_{\tau\in\mathcal{S}}\E\left\{e^{-r\tau}(f(D_\tau;\,p) - c)\,|\, D_0 = y\right\},
\end{align}
where $f(\cdot;~\cdot)$ is defined in \eqref{F1}.

\end{proposition}

\begin{proof}Let us consider a transformation of ${V}(D_0;p)$. First, by rearranging of
the first integral in \eqref{cancellableV} and using $\1_{\{\tau\geq \tdk\}} =
\1- \1_{\{\tau< \tdk\}}$, we obtain
\begin{align}{V}(y;\,p) &=    \E\left\{- \int_0^{\tdk }e^{-rt}p\, dt + \alpha e^{-r\tdk}\,|\,D_0=y\right\}\notag\\
& +    \underbrace{\sup_{0\leq \tau<\infty}
\E\left\{ \int_{\tdk \wedge \tau}^{\tdk}e^{-rt}p\, dt - c e^{-r\tau} \1_{\{\tau< \tdk\}}
  - \alpha e^{-r\tdk} \1_{\{\tau< \tdk\}}\,|\, D_0 = y\right\}}_{=:G(y;\,p)}\notag\\
&=-f(y;\,p)+G(y;\,p). \label{f_g}
\end{align}
Note that the first term is explicitly given in \eqref{F1} and \eqref{zdprop}, and it does not depend on $\tau$. Since the second term
depends on $\tau$ only through its truncated counterpart $\tau \wedge \tdk \leq \tdk$, and that $\tau = \tau_D(k)$ is suboptimal,  we can in fact
consider maximizing over the restricted collection of stopping times
$\mathcal{S}=  \{\tau \in\F\,:\, 0\leq \tau < \tdk\}$. As a result, the second term simplifies to
\begin{align}\label{cancellableg}G(y;p) &=    \sup_{  \tau\in \mathcal{S}}
\E\left\{  \int_{ \tau}^{\tdk}e^{-rt}p\, dt - c e^{-r\tau} \1_{\{\tau< \tdk\}}
- \alpha e^{-r\tdk} \1_{\{\tau<\tdk\}} \,|\,D_0=y\right\}. \notag
\end{align}

Then, using the fact that $\{ \tau < \tdk, \tau <\infty\}  = \{ D_\tau< k, \tau <\infty\}$, as well as the strong Markov property of $X$,
we can  write
\[G(y;\,p) =    \sup_{  \tau\in \mathcal{S}}
\E\left\{  e^{-r \tau} \tilde{f}(D_\tau;p)  \1_{\{\tau<\infty\}}\,|\,D_0=y\right\}, \notag
\]
where
\begin{align}\label{hy}\tilde{f}(y;p) = \1_{\{y< k\}} \E\left\{ \int_0^{\tau_D(k) }e^{-rt}p\, dt - \alpha
e^{-r\tau_D(k)}-c \,|\,D_\tau=y\right\} .\end{align}
Hence, we complete the proof by simply noting that  $\tilde{f}(y;\,p) = f(y;\,p) - c$ (compare \eqref{hy} and \eqref{F1}).
\end{proof}

Using this decomposition, we can determine the optimal cancellation strategy from  the optimal stopping problem $G(y)$, which we will solve explicitly in the next subsection.

\subsection{Optimal Cancellation Strategy}\label{sect-optcancel}

In order to determine the  optimal cancellation strategy for $V(y;\,p)$ in \eqref{decomp}, it is sufficient to solve the optimal stopping problem represented by $g$ in \eqref{f_g} for a fixed $p$. To simplify notations, let us denote by $f(\cdot)=f(\cdot;\,p)$ and $\tilde{f}(\cdot)=\tilde{f}(\cdot;\,p)$. Our method of solution  consists of  two main steps:
\begin{enumerate} \item We conjecture a candidate class of  stopping times defined by  $\tau\mathop{:=}\tau_D^-(\theta)\wedge\tau_D(k)\in\mathcal{S}$, where
\begin{eqnarray}
\tau_D^-(\theta)=\inf\{t\ge 0\,:\,D_t\le \theta\}, \quad 0<\theta<k. \label{candidate}
\end{eqnarray}
This leads us to look for a candidate optimal threshold $\theta^*\in(0,k)$ using the principle of \emph{smooth pasting} (see \eqref{pasting}).

    \item We rigorously verify via a martingale argument that  the cancellation strategy based on the threshold $\theta^*$ is indeed optimal. \end{enumerate}

\textbf{Step 1.} From the properties of Laplace function $\xi(\cdot)$ (see Lemma \ref{lem1} below), we know the reward function $\tilde{f}(\cdot) := f(\cdot) - c$ in \eqref{decomp}  is a decreasing concave. Therefore, if $\tilde{f}(0)\le 0$, then the second term of \eqref{decomp} is non-positive, and  it is optimal for the protection buyer to never cancel the insurance, i.e., $\tau=\infty$.   Hence, in search of nontrivial optimal exercise strategies,   it is sufficient to study only the case with $\tilde{f}(0)>0$, which is equivalent to
\begin{align} p > \frac{ r (c + \alpha \xi(0))}{1 - \xi(0)}\ge 0.\end{align}

For each  stopping rule conjectured in \eqref{candidate},  we compute explicitly the second term of \eqref{decomp} as
\begin{align}\label{valfunction}  g(y;\theta)&\mathop{:=}\E\left\{e^{-r(\tau_D^-(\theta)\wedge \tau_D(k))}\tilde{f}(D_{\tau_D^-(\theta)\wedge \tau_D(k)})\,|\,D_0=y\right\} \\
&\mathop{=}\E\{e^{-r\tau_D^-(\theta)}\1_{\{\tau_D^-(\theta)<\tau_D(k)\}}\tilde{f}(\theta)\,|\,D_0=y\} +\E\{e^{-r\tau_D(k)}\1_{\{\tau_D(k)\le\tau_D^-(\theta)\}}\tilde{f}(k)\,|\,D_0=y\}\notag \vspace{3pt}\\
  &=\left\{\displaystyle\begin{array}{ll}e^{\frac{\mu}{\sigma^2}(y-\theta)}\frac{\sinh(\Xi_{\mu,\sigma}^r(k-y))}{\sinh(\Xi_{\mu,\sigma}^r(k-\theta))}\tilde{f}(\theta),~&\text{if}~y>\theta\vspace{3pt}\\
   \tilde{f}(y), &\text{if}~y\le\theta \end{array}.\right.
\end{align}
The \emph{candidate optimal} cancellation threshold $\theta^*\in(0,k)$ is found from  the smooth pasting condition:
\begin{eqnarray}\label{pasting}
\frac{\partial}{\partial y}\bigg|_{y=\theta}g(y; \theta)=\tilde{f}^{'}(\theta).
\end{eqnarray}
This is equivalent to seeking the root  $\theta^*$ of the equation:
\begin{align}\label{Ftheta}F(\theta)\mathop{:=}\bigg(\frac{\mu}{\sigma^2}-\Xi_{\mu,\sigma}^r\coth(\Xi_{\mu,\sigma}^r(k-\theta))\bigg)\tilde{f}(\theta)-\tilde{f}^{'}(\theta) =0, \end{align}
where $\tilde{f}$ and $\tilde{f}'$ are explicit in view of  \eqref{F1} and \eqref{zdprop}. Next, we show that the root $\theta^*$ exists and is unique (see Section \ref{proof-prop-unique} for the proof).

\begin{lemma} \label{prop-unique} There exists a unique $\theta^*\in(0,k)$ satisfying the smooth pasting condition \eqref{pasting}.
\end{lemma}

\textbf{Step 2.} With the candidate optimal threshold  $\theta^*$ from (\ref{pasting}), we  now verify that  the candidate value function $g(y;\theta^*)$ dominates the reward function $\tilde{f}(y) = f(y) - c$. Recall that $g(y;\theta^*)=\tilde{f}(y)$ for $y\in(0,\theta^*)$.

\begin{lemma}\label{lem2}The  value function corresponding to the candidate optimal threshold $\theta^*$ satisfies
  \[g(y;\theta^*)>\tilde{f}(y),\quad \forall y\in(\theta^*,k).\]
\end{lemma}

  We provide a proof in  \ref{sect-appx-gf}.  By the definition of  $g(y;\theta^*)$ in \eqref{valfunction},  repeated conditioning yields that the  stopped process  $\{e^{-r ( t \wedge \tau_D^-(\theta^*) \wedge  \tau_D(k)) }g(D_{t \wedge \tau_D^-(\theta^*)\wedge  \tau_D(k))};\theta^*)\}_{t\ge 0}$ is a martingale. For $y\in[0,\theta^*)$, we have
\[
\frac{1}{2}\sigma^2\tilde{f}^{''}(y)-\mu \tilde{f}^{'}(y)-r\tilde{f}(y)=-C\bigg(\frac{1}{2}\sigma^2\xi^{''}(y)-\mu\xi^{'}(y)-r(\xi(y)-\xi(\theta_0))\bigg)=-Cr\xi(\theta_0)<0,\]
where $C=\alpha+\frac{p}{r}$.
As a result, the stopped process $\{e^{-r (t\wedge  \tau_D(k))}g(D_{t\wedge  \tau_D(k))};\theta^*)\}_{t\ge 0}$ is in fact a super-martingale.

To finalize the proof, we note that for $y \in (\theta^*, k)$ and any stopping time $\tau \in \mathcal{S}$,
\begin{align}\label{supermtg}g(y;\theta^*) \ge \E\{e^{-r\tau }g(D_{\tau};\theta^*)\,|\,D_0=y\} \ge   \E \{e^{-r\tau }\tilde{f}(D_{\tau})\,|\,D_0=y\}.\end{align} Maximizing over $\tau$, we see that $g(y;\theta^*)\ge G(y)$.   On the other hand, \eqref{supermtg} becomes an equality when  $\tau = \tau^-_D(\theta^*)$, which yields the reverse inequality $g(y;\theta^*)\le G(y)$.  As a result, the stopping time $\tau^-_D(\theta^*)$ is indeed the solution to the optimal stopping problem $G(y)$.

In summary, the protection buyer will continue to pay the premium over time until the drawdown  process $D$ either falls to the level $\theta^*$  in \eqref{pasting} or reaches to the level $k$ specified by the contract, whichever comes first. In Figure \ref{fig1} (left), we illustrate the optimal cancellation level $\theta^*$. As shown in our proof, the optimal stopping value function $g(y)$  connects smoothly with the intrinsic value $\tilde{f}(y)=f(y)-c$  at $y=\theta^*$. In Figure \ref{fig1} (right), we show that the fair premium $P^*$ is decreasing with respect to the protection downdown size $k$. This is intuitive since the drawdown time $\tau_D(k)$ is almost surely longer for a larger drawdown size $k$ and the payment at $\tau_D(k)$ is fixed at $\alpha$. The protection buyer is expected to pay over a longer period of time but at a lower premium rate.

    \begin{figure}[th!]\begin{center}  \subfigure[Smooth pasting]{\label{A1}
          \includegraphics[width=0.48\textwidth]{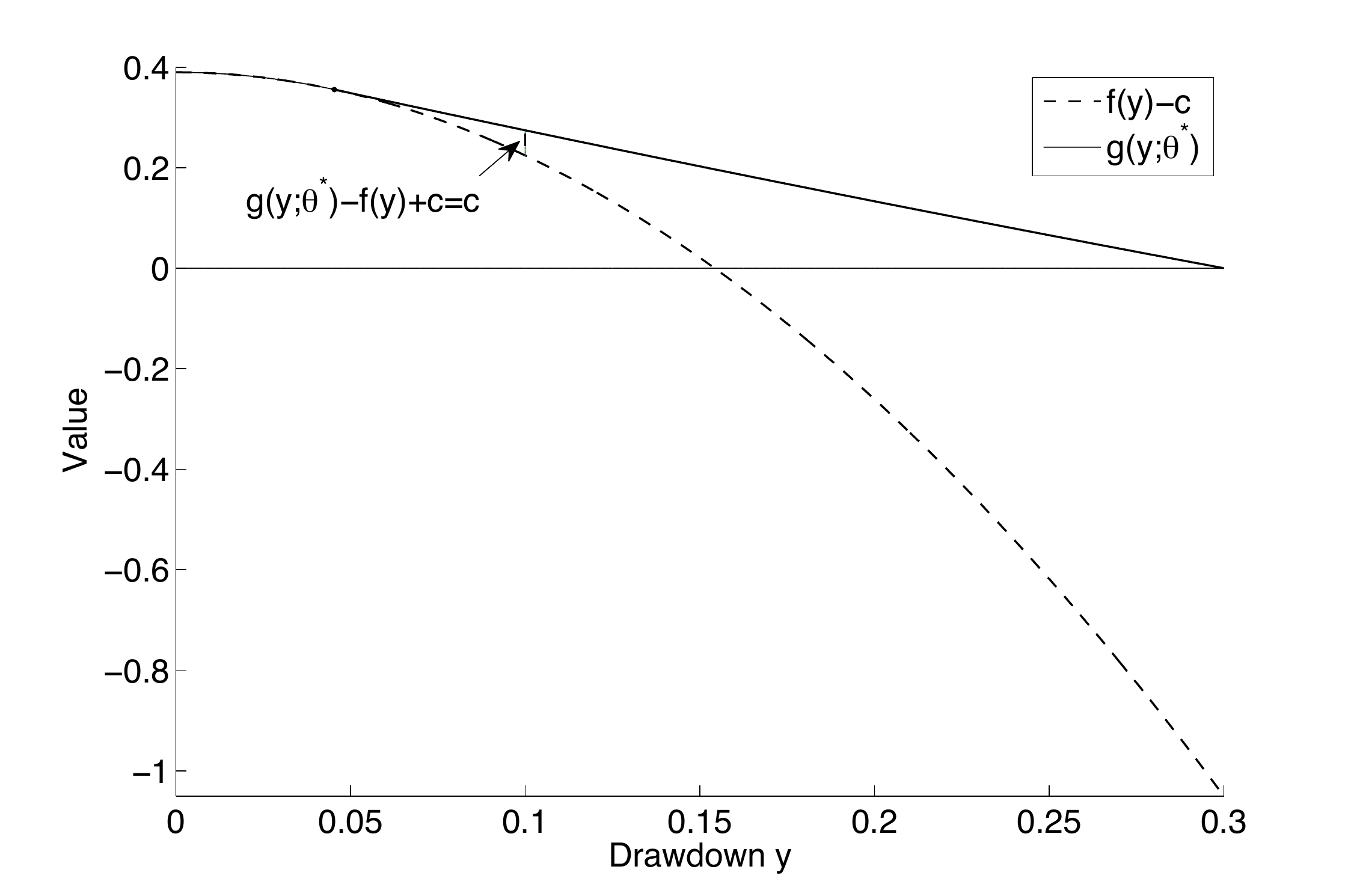}
        }
        \subfigure[Fair premium vs. $k$]{
          \label{A2}
          \includegraphics[width=0.48\textwidth]{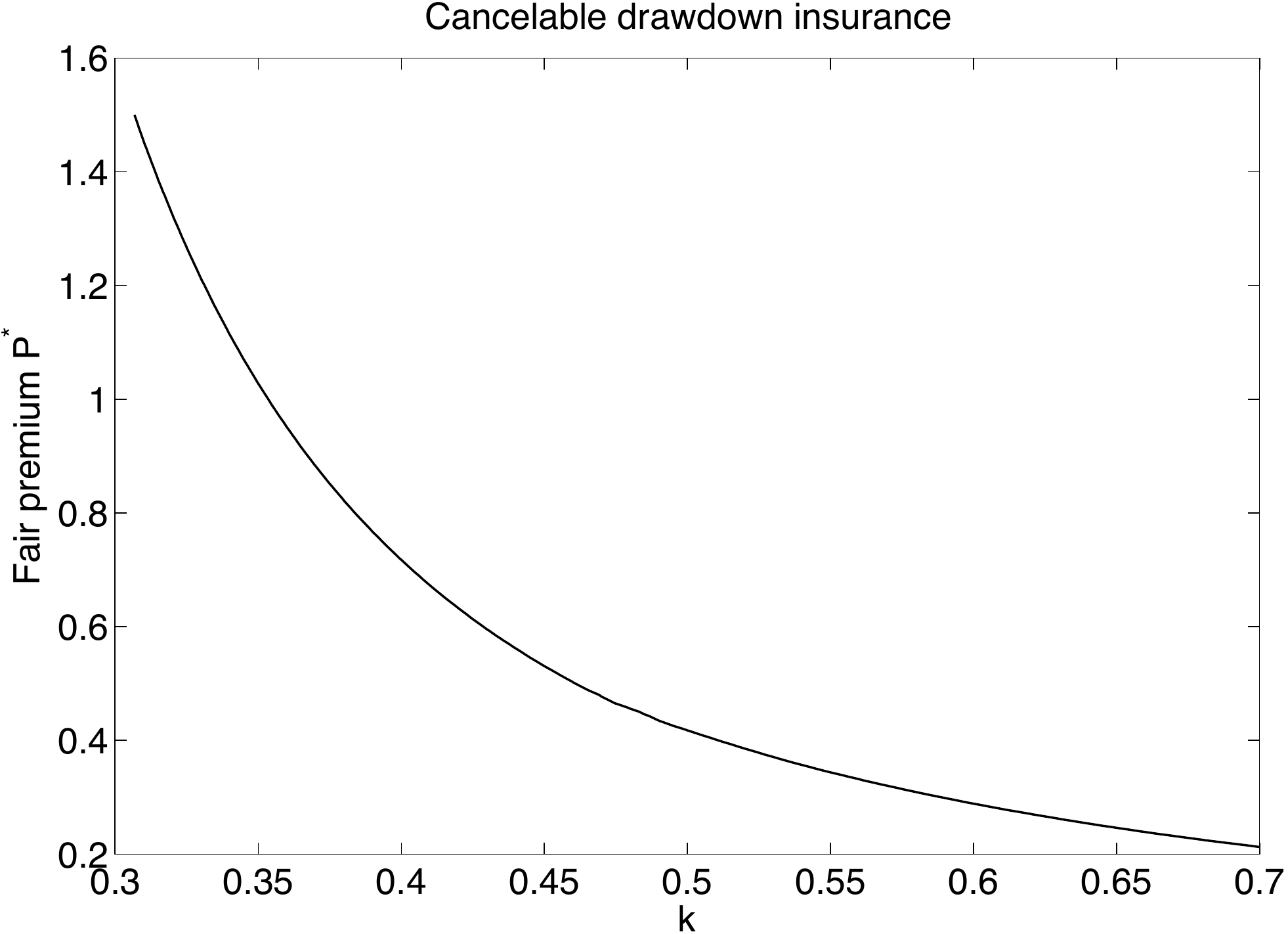}
        }
      \begin{small}
      \caption{Left panel:  the optimal stopping value function $g$ (solid) dominates  and pastes smoothly onto the intrinsic value  $\tilde{f}=f-c$ (dashed). It is  optimal to cancel the insurance as soon as the drawdown process falls to   $\theta^* = 0.05$ (at which $g$ and $\tilde{f}$ meets).  The parameters are $r=2\%, \sigma=30\%, y=0.1, k=0.3, c=0.05$,  $\alpha=1$, and $p$ is taken to be the fair premium value $P^*=1.5245$. At $y=0.1$, according to the fair premium equation $V(y\,;P^*)=0$, and hence $g (y; \theta^*)= f(y)$ here. Right panel: the fair premium of the cancellable drawdown insurance decreases with respect to the drawdown level $k$ specified for the contract.}\label{fig1} \end{small}     \end{center}
\end{figure}

Lastly, with the optimal cancellation strategy, we can also compute the expected time to contract termination, either as a result of a drawdown or voluntary cancellation. Precisely, we have
\begin{proposition} For $0<\theta^*<y<k$, we have
  \begin{align}\label{ExpTime}\E_{\P}\{\tau_D^{-}(\theta^*)\wedge\tdk|D_0=y\}=\frac{(y-\theta^*)\rho_{\tau}(y-\theta^*;\,k-\theta^*)+(y-k)e^{\frac{2\tilde{\mu}}{\sigma^2}(y-k)}\rho_{\tau}(k-y;\,k-\theta^*)}{\tilde{\mu}},\end{align}
\end{proposition}
where $\rho_{\tau}(\cdot,\cdot)$ is defined in Proposition \ref{exptime}.
\begin{proof}
  According to the optimal cancellation strategy, we have
  \[\tau_D^{-}(\theta^*)\wedge\tdk=\tau_{x+y-\theta^*}\wedge\tau_{y-k},~\P\text{-a.s.}\]
  where $\tau_w=\inf\{t\ge0\,:\,X_t=w\}$. Applying the optional sampling theorem to the uniformly martingale $(M_{t\wedge\tau_{x+y-\theta^*}\wedge\tau_{y-k}})_{t\ge0}$ with $M_t=X_t-\tilde{\mu}t$ if $\tilde{\mu}\neq0$, or $M_t=(X_t)^2-\sigma^2t$ if $\tilde{\mu}=0$, we obtain the result in the \eqref{ExpTime}.
  \end{proof}

\begin{remark}\label{remark-finitemat}
In the finite maturity case, the set of candidate stopping times  is changed to $\{\tau \in \mathbb{F}\,:\, 0\le t\le T\}$ in \eqref{cancellableV}.  Like Proposition \ref{decomp}, the contract value ${V}_T(y;p)$ at time zero for premium rate $p$  still admits the   decomposition
  \[{V}_T(y;p)=-f_T(0,y;p)+\sup_{0\le \tau\le T}\E\{e^{-r\tau}({f}_T(\tau,D_\tau;p) - c)\1_{\{\tau<\tdk\}}\,|\,D_0=y\},\]
  where \[f_T(t,y)= \frac{p}{r}-(\alpha+\frac{p}{r})\,\xi_T(t,y), \]
and $\xi_T(t,y)$ is  the conditional Laplace transform of $\tdk\wedge(T-t)$:
  \[\xi_T(t,y)=\E\{e^{-r(\tdk\wedge(T-t))}\,|\,D_t=y\}, \quad 0\le t\le \tdk\wedge T.\]
This problem is no longer time-homogeneous, and the fair premium can be determined by numerically solving the associated optimal stopping problem.
\end{remark}

\section{Incorporating Drawup Contingency}\label{sect-drawuptoo}
We now consider an  insurance contract that provides  protection  from any specified drawdown with a drawup contingency.  This insurance contract  may expire early, whereby stopping the premium payment, if a drawup event occurs prior to drawdown or maturity.  From the investor's viewpoint, the realization of a drawup implies little need to insure against a drawdown. Therefore, this drawup contingency   is an attractive cost-reducing  feature to the investor.

\subsection{The Finite-Maturity Case}
 First, we consider the case with a finite maturity $T$.  Specifically, if a $k$-unit  drawdown  occurs prior to a drawup of the same size or the expiration date $T$, then the investor will receive the insured amount $\alpha$ and stop the premium payment thereafter. Hence, the risk-neutral discounted expected payoff   to the investor is given by
\begin{align}\label{1}v (y,z; \,p) = \E^{y,z}\left\{- \int_0^{\tdk \wedge \tuk \wedge T}e^{-rt}p dt + \alpha e^{-r\tdk} \1_{\{\tdk\le \tuk\wedge T\}}\right\},
\end{align}
where the expectation is taken under the pricing measure $\Q^{y,z}(\cdot)\equiv\Q(\cdot\,|\,D_0=y, U_0=z)$.

The fair premium $P^*$ is chosen such that the contract has value zero at time zero, that is,  \begin{align}v(P^*)=0.\label{fairp}\end{align} Applying \eqref{fairp} to \eqref{1},  we obtain a formula for the fair premium: \begin{eqnarray}\label{pstar}P^* = \frac{r \alpha \E^{y,z}\{e^{-r\tdk} \1_{\{\tdk\le \tuk\wedge T\}}\} }{1-\E^{y,z}\{e^{-r ({\tdk
\wedge\tuk \wedge T})}\}}.\end{eqnarray}
As  a result, the fair premium involves computing the  expectation \mbox{$\E^{y,z}\{e^{-r\tdk} \1_{\{\tdk\le \tuk\wedge T\}}\}$} and the Laplace transform  of \mbox{${\tdk
\wedge\tuk \wedge T}$}.

In order to determine the fair premium $P^*$ in \eqref{pstar}, we first write
\begin{align}
  L_r^{T}&:=  \E^{y,z} \{e^{-r\tdk}\1_{\{\tdk\le \tuk\wedge T\}}\}\label{LmuT}\\
  &=\int_0^Te^{-rt}\,\frac{\partial}{\partial t} \Q^{y,z}(\tdk < \tuk \wedge t)\,dt.\label{truncatedL}
  \end{align}
 The special case of the probability on the right-hand side, $\Q^{0,0}$ is derived using results from \cite{ZhanHadj} (eq. (39)-(40)), namely,
\begin{align}
  \Q^{0,0}(\tdk < \tuk \wedge t)=&\frac{e^{-\frac{2\mu k}{\sigma^2}}+\frac{2\mu k}{\sigma^2}-1}{(e^{-\frac{\mu k}{\sigma^2}}-e^{\frac{\mu k}{\sigma^2}})^2}-\sum_{n=1}^\infty \frac{2n^2\pi^2}{C_n^2}\exp\bigg(-\frac{\sigma^2C_n}{2k^2}t\bigg)\notag \\
  &\times\bigg[(1-(-1)^ne^{-\frac{\mu k}{\sigma^2}})\bigg(1+\frac{n^2\pi^2\sigma^2 t}{k^2}-\frac{4\mu^2k^2}{\sigma^4C_n}\bigg)+(-1)^n\frac{\mu  k}{\sigma^2} e^{-\frac{\mu k}{\sigma^2}}\bigg], \label{Q1}
  \end{align}
  where $C_n=n^2\pi^2+\mu^2k^2/\sigma^4$. Therefore, the expectation \eqref{truncatedL} can be computed via numerical integration. In the general case that $y\vee z>0$, we have the following result.
  \begin{proposition}
    In the model \eqref{DBM}, for $0\le y,z<k$ and $y\vee z>0$, we have
    \begin{align}
\Q^{y,z}(\tdk\in dt,\tuk>t)&=F_y^\mu(t)dt+G_z^\mu(t)dt-G_{k-y}^\mu(t)dt,\label{Qdensity}
\end{align}
where
\begin{align}
  F_y^\mu(t)&\mathop{:=}\frac{\sigma^2}{k^2}\sum_{n=1}^\infty(n\pi)e^{\frac{(y-k)\mu}{\sigma^2}}\exp\bigg(-\frac{
    \sigma^2C_n}{2k^2}t\bigg)\sin\frac{n\pi(k-y)}{k},\\
  G_z^\mu(t)&\mathop{:=}\frac{\sigma^2}{k^2}\sum_{n=1}^\infty(n\pi)e^{-\frac{\mu z}{\sigma^2}}\exp\bigg(-\frac{\sigma^2C_n}{2k^2}t\bigg)\times\bigg\{\frac{n^2\pi^2\sigma^2t-2k^2}{C_nk}\bigg(\frac{n\pi}{k}\cos\frac{n\pi z}{k}+\frac{\mu}{\sigma^2}\sin\frac{n\pi z}{k}\bigg)\notag\\
  &+\frac{n\pi}{C_n}\bigg[\frac{n\pi}{k}\bigg(\frac{2k^2\mu}{C_n\sigma^2}+{z}\bigg)\sin\frac{n\pi z}{k}+\bigg(1-\frac{\mu z}{\sigma^2}-\frac{2\mu^2k^2}{C_n\sigma^4}\bigg)\cos\frac{n\pi z}{k}\bigg]\bigg\}.
  \end{align}
\end{proposition}
\begin{proof}We begin by differentiating both sides of (2.7)  in \cite{CarrZhanHaji} with respect to maturity $t$ to obtain that
  \begin{align}
\Q^{y,z}(\tdk\in dt,\tuk>t)=q(t,x,y+x-k, y+x)dt+\bigg(\int_{y+x-k}^{x-z}\frac{\partial}{\partial k}q(t,x,u,u+k)du\bigg)dt,\label{FiniteM}
\end{align}
where
\[q(t,x,u,u+k)dt=\Q^{y,z}(\tau_u\in dt, \tau_{u+k}>t)\]{with} $\tau_w:=\inf\{t\ge0\,:\,X_t=w\}$ for $w\in \{u,u+k\}$. The function $q$,  derived in Theorem 5.1 of \cite{Anderson60}, is given by\[q(t,x,u,u+k)=\frac{\sigma^2}{k^2}\sum_{n=1}^\infty(n\pi)e^{\frac{\mu(u-x)}{\sigma^2}}\exp\bigg(-\frac{\sigma^2C_n}{2k^2}t\bigg)\sin\frac{n\pi(x-u)}{k}.\]
Integration yields \eqref{Qdensity} and this completes the proof.
  \end{proof}
Similarly, we express the Laplace transform of  $\tdk \wedge \tuk \wedge T$ as
    \begin{align}\label{E2}
\E^{y,z}\{e^{-r(\tdk \wedge \tuk \wedge T)}\}=&-\int_0^Te^{-rt}\frac{\partial}{\partial t}\Q^{y,z}(\tdk \wedge \tuk >t)dt.
\end{align}
To compute this, we notice that the equivalence of the probabilities (under the reflection of the processes $(X, \overline{X}, \underline{X})$ about $x$): \begin{eqnarray}\label{REFLECTION}\Q^{y,z}_\mu(\tuk\in dt, \tdk>t)=\Q^{z,y}_{-\mu}(\tdk\in dt, \tuk >t).\end{eqnarray}
Therefore, we have
  \begin{align}
-&\frac{\partial}{\partial t}\Q^{y,z}_\mu(\tdk \wedge \tuk >t)dt\notag\\
    &=\Q^{y,z}_\mu(\tdk\in dt,\tuk>t)+\Q^{y,z}_\mu(\tuk\in dt, \tdk> t)\notag\\
    &=F_y^\mu(t)dt+G_z^\mu(t)dt-G_{k-y}^\mu(t)dt+F_{ z}^{-\mu}(t)dt+G_{y}^{-\mu}(t)dt-G_{k- z}^{-\mu}(t)dt.\label{QQQ}
    \end{align} where  $\Q^{y,z}_{\mu}$ denotes the pricing measure whereby  $X$ has drift $\mu$.
Hence, we can again compute the Laplace transform of  $\tdk \wedge \tuk \wedge T$ by numerical integration, and obtain the fair premium $P^*$ for the drawdown insurance via \eqref{pstar}.
\vspace{5pt}

\begin{remark} The expectation \mbox{$\E^{y,z}\{e^{-r\tdk} \1_{\{\tdk\le \tuk\wedge T\}}\}$} and the Laplace transform  of \mbox{${\tdk
\wedge\tuk \wedge T}$} are in fact linked. This is seen through \eqref{REFLECTION}:
\[\E^{y,z}\{e^{-r(\tdk \wedge \tuk \wedge T)}\}=L_r^{T}+R_{r}^{T},\]
where $L_r^T$ is the expectation defined in \eqref{LmuT}, and
\begin{eqnarray}R_r^T\mathop{:=}\E^{y,z}\{e^{-r\tdk}\1_{\{\tdk\le\tuk\wedge T\}}\}.\label{Rr}
  \end{eqnarray}
\end{remark}
\vspace{5pt}

\begin{remark} If the protection buyer pays a periodic premium at
times $t_i=i \Delta t$, $i=0, \ldots, n-1$, with $\Delta t = T/n$, then the fair premium is
\begin{align} {p}^{(n)*}=\frac{ \alpha\,  \E^{y,z}\left\{ e^{-r\tdk} \1_{\{\tdk\le \tuk\wedge T\}} \right\}}{ \sum_{i=0}^{n-1} e^{-r t_i }  \Q^{y,z}\{\tdk \wedge \tuk > t_i\}}.
\end{align}
Compared to the continuous premium case, the fair premium ${p}^{(n)*}$ here involves a sum of the probabilities: $\Q^{y,z}\{\tdk \wedge \tuk > t_i\}$, each given by \eqref{QQQ} above.
\end{remark}

\subsection{Perpetual Case}Now, we consider the drawdown insurance contract that will expire not at a fixed finite time $T$ but as soon as a drawdown/drawup of size $k$ occurs. To study this perpetual case, we take  $T=\infty$ in \eqref{1}. As the next proposition shows, we have  a simple closed-form solution for the fair premium $P^*$, allowing for instant computation of the fair premium and amenable for sensitivity analysis.

\begin{proposition} \label{prop-perp1}The perpetual drawdown insurance  fair premium $P^*$ is given by
\begin{align}
  P^*=\frac{r\alpha L_r^{\infty}}{1-L_r^{\infty}-R_r^{\infty}},
\end{align}
where
\begin{align}\label{L}L_r^{\infty}=F_{y}^\mu+G_{ z}^{\mu}-G_{k-y}^{\mu},\quad R_{r}^{\infty}=F_{ z}^{-\mu}+G_{y}^{-\mu}-G_{k- z}^{-\mu},
\end{align}
with
\begin{align}
F_{y}^{\mu}\mathop{:=}e^{\frac{\mu}{\sigma^2}(y-k)}\frac{\sinh(y\s)}{\sinh(k\s)},~G_{ z}^\mu\mathop{:=}\frac{\s}{2r/\sigma^2}\frac{e^{-\frac{\mu}{\sigma^2} z}\left(-\frac{\mu}{\sigma^2}\sinh( z\s)-\s\cosh( z\s)\right)}{\sinh^2(k\s)}.
  \end{align}
\end{proposition}

\begin{proof}
  In the perpetual case, the fair premium is given by
  \[P^*=\frac{r\alpha\E^{y,z}\{e^{-r\tdk}\1_{\{\tdk<\tuk\}}\}}{1-\E^{y,z}\{e^{-r(\tdk\wedge\tuk)}\}}=\frac{r\alpha L_r^{\infty}}{1-L_r^{\infty}-R_r^{\infty}}.\]
  where $L_{r}^{\infty}=\E^{y,z}\{e^{-r\tau_D(k)}\1_{\{\tdk<\tuk\}}\}$ and $R_{r}^{\infty}=\E^{y,z}\{e^{-r\tdk}\1_{\{\tuk<\tdk\}}\}$.
  To get formulas for $L_r^\infty$ and $R_r^\infty$, we begin by multiplying both sides of \eqref{FiniteM} by $e^{-rt}$ and integrate out $t\in[0,\infty)$. Then we obtain that
  \[L_r^\infty=\E^{y,z}\{e^{-r\tau_{y+x-k}}\1_{\{\tau_{y+x-k}<\tau_{y+x}\}}\}+\int_{y+x-k}^{x-z}\frac{\partial}{\partial k}\E^{y,z}\{e^{-r\tau_u}\1_{\{\tau_u<\tau_{u+k}\}}\}du,\]
  where $\tau_w=\inf\{t\ge0\,:\,X_t=w\}$.
  Using  formulas in \cite[p.295]{Borodin1996}, we have that for $u\le x-z<x+y\le u+k$
  \begin{align*}\E^{y,z}\{e^{-r\tau_u}\1_{\{\tau_u<\tau_{u+k}\}}\}&=e^{\frac{\mu}{\sigma^2}(u-x)}\frac{\sinh((u+k-x)\Xi_{\mu,\sigma}^r)}{\sinh(k\Xi_{\mu,\sigma}^r)},\\
    \frac{\partial}{\partial k}\E^{y,z}\{e^{-r\tau_u}\1_{\{\tau_u<\tau_{u+k}\}}\}&=\s e^{\frac{\mu}{\sigma^2}(u-x)}\frac{\sinh((x-u)\s)}{\sinh^2(k\s)}.
  \end{align*}
  An integration yields $L_r^\infty$.
  The computation of $R_r^{\infty}$ follows from the discussion in the proof of Proposition 3.1. This completes the proof of the proposition.
\end{proof}
\vspace{5pt}

Finally, the probability that a drawdown is realized prior to a drawup, meaning that the protection amount will be paid to the buyer before the contract expires upon drawup, is given by
\begin{proposition} Let $y,z\ge0$ such that $y+z<k$, then
  \begin{align}\P(\tdk<\tuk|D_0=y,U_0=z)=e^{\frac{2\tilde{\mu}}{\sigma^2}(y-k)}\rho_{\tau}(k-y;\,k)+\frac{e^{-\frac{2\tilde{\mu} }{\sigma^2}(k-y)}+\frac{2\tilde{\mu}}{\sigma^2}(k-y-z)-e^{-\frac{2\tilde{\mu}}{\sigma^2}z}}{4\sinh^2(\frac{\tilde{\mu}}{\sigma^2}k)}, \label{Ptdk}\end{align}
where $\rho_{\tau}(\cdot;\cdot)$ is defined in Proposition \ref{exptime}.
\end{proposition}
\begin{proof}
  From the proof of Proposition \ref{prop-perp1}, we obtain that
  \begin{align*}\P(\tdk<\tuk|D_0=y,U_0=z)=&\lim_{r\to0^+}(F_y^{\tilde{\mu}}+G_z^{\tilde{\mu}}-G_{k-y}^{\tilde{\mu}})\\=&e^{\frac{2\tilde{\mu}}{\sigma^2}(y-k)}\rho_{\tau}(k-y;\,k)+\lim_{r\to0^+}(G_z^{\tilde{\mu}}-G_{k-y}^{\tilde{\mu}}).\end{align*}
Finally, L'H\^opital's rule yields the last limit and \eqref{Ptdk}.
\end{proof}

 In Figure \ref{fig2} (left), we see that  the fair premium increases with the maturity $T$, which is due to the higher likelihood of the drawdown event at or before expiration. For the perpetual case, we illustrate in   Figure \ref{fig2} (right) that higher volatility leads to higher fair premium. From this observation, it is expected in a volatile market drawdown insurance would become more costly.

 \begin{figure}[ht!]\begin{center}  \subfigure[Fair premium vs. Maturity $T$]{\label{B1}
          \includegraphics[width=0.475\textwidth]{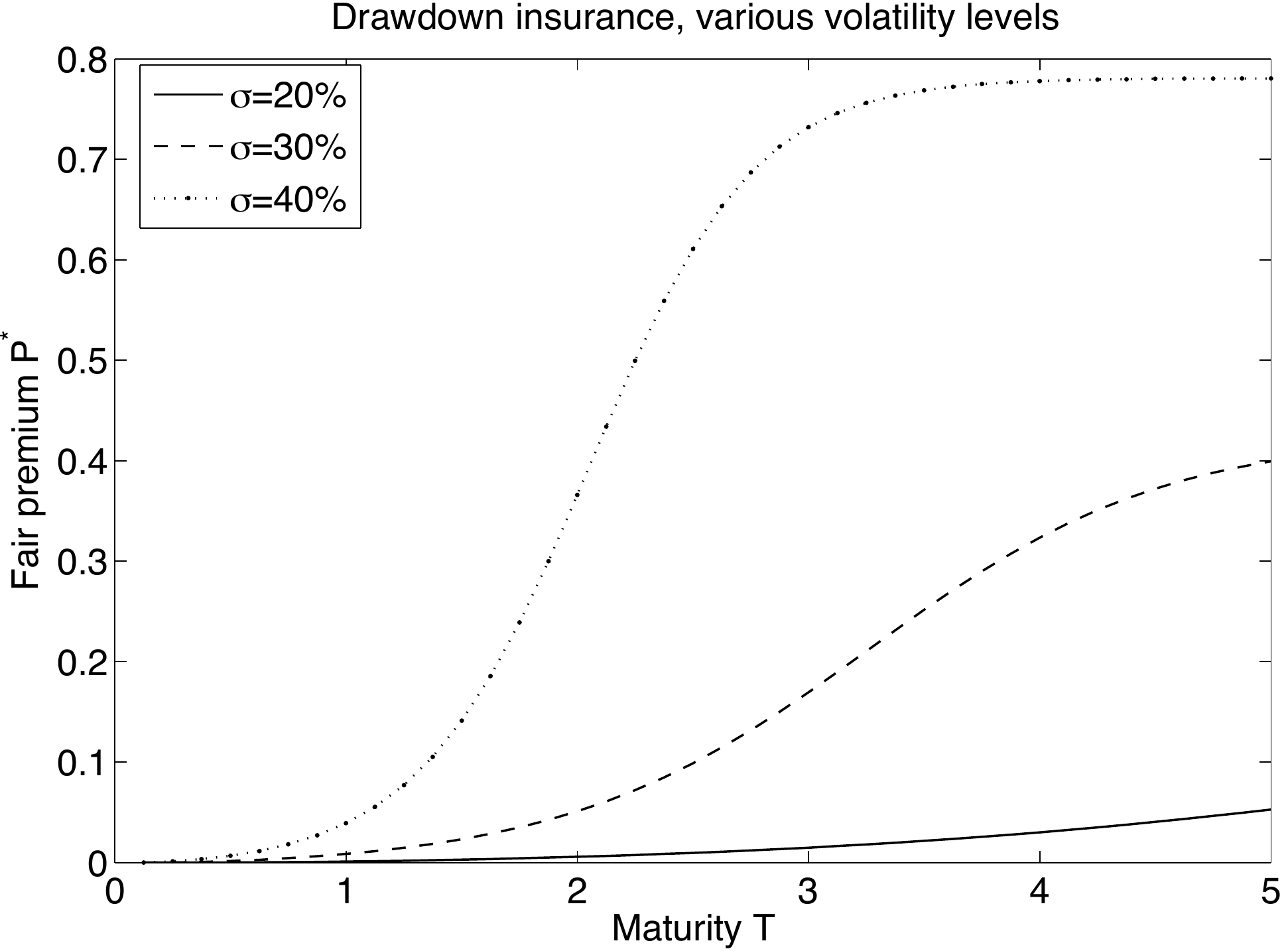}
        }
        \subfigure[(Perpetrual) fair premium vs. volatility $\sigma$]{
          \label{B2}
          \includegraphics[width=0.481\textwidth]{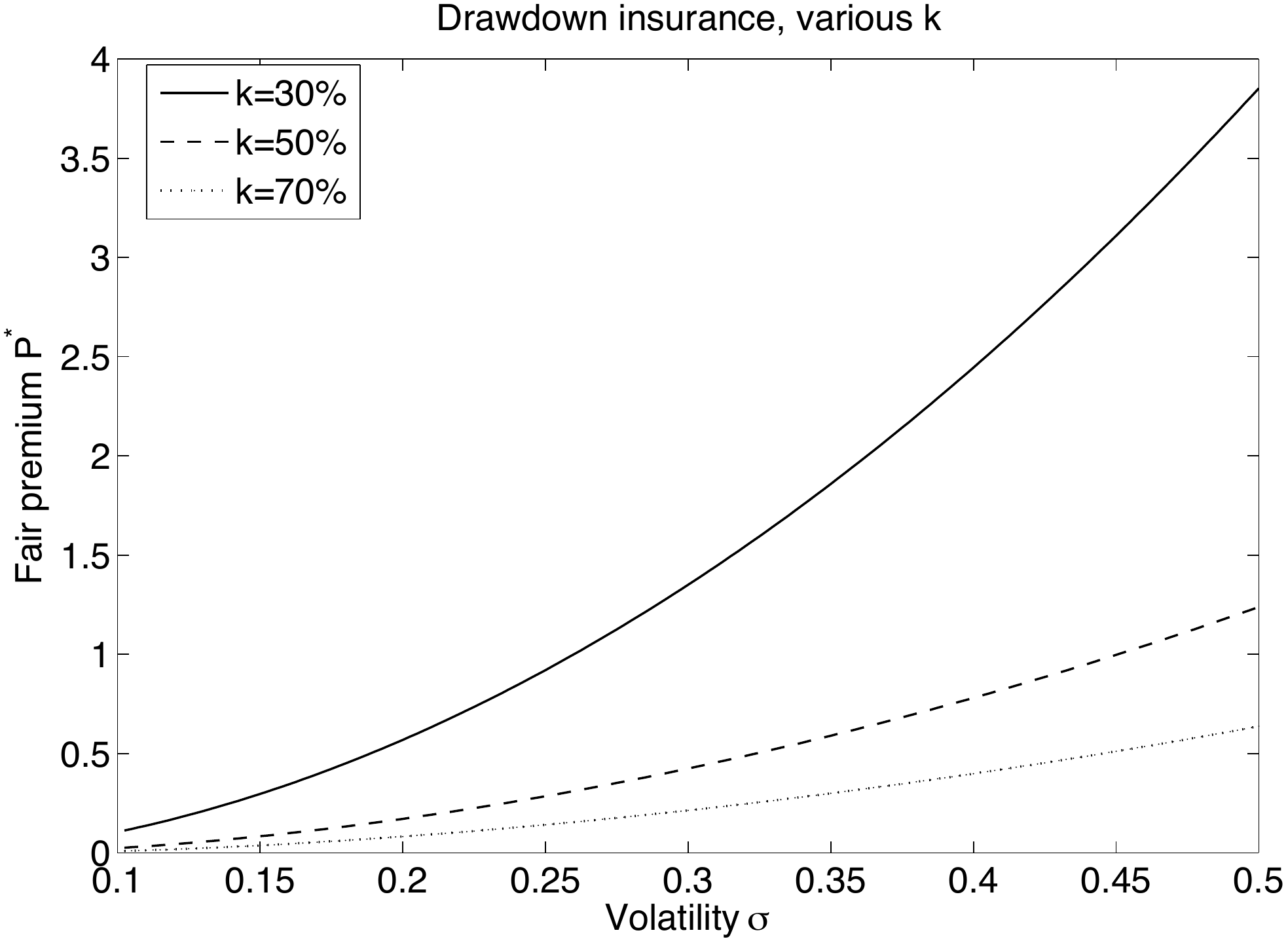}
        }
      \end{center}
      \begin{small}
      \caption{Left panel: the fair premium of a drawdown insurance is increasing with the maturity $T$. Right panel: the fair premium of a perpetual drawdown insurance also increases with volatility $\sigma$. The parameters are $r=2\%,  y=z=0.1$, $\alpha=1$, and $k=50\%$ (left).  }
    \label{fig2}  \end{small}
\end{figure}

\section{Drawdown Insurance on a Defaultable Stock}\label{sect-def}
In contrast to  a market index,  an individual stock may experience a large drawdown through continuous downward movement or  a sudden default event. Therefore, in  order to insure against the drawdown of the stock, it is useful  to incorporate the default risk into the stock price dynamics. To this end, we extend our analysis to a stock with reduced-form (intensity based) default risk.

Under the risk-neutral  measure $\Q$, the defaultable stock price $\tilde{S}$ evolves according to
\begin{equation}d\tilde{S}_t = (r + \lambda )\tilde{S}_t \, dt + \sigma \tilde{S}_t\, dW_t - \tilde{S}_{t-}\,dN_t,\qquad \tilde{S}_0=\tilde{s}>0,\label{s_def}
\end{equation}where  $\lambda$ is the constant default intensity for the single jump process $N_t=1_{\{ t \geq \zeta\}}$, with $\zeta \sim \exp(\lambda)$ independent of   the Brownian motion $W$ under $\Q$. At $\zeta$, the stock price immediately drops to zero and remains there permanently, i.e.\ for a.e.~$\omega \in \Omega$, $\tilde{S}_{t}(\omega) = 0,  \forall t \ge \zeta(\omega)$. Similar equity models have been considered e.g.\ in \cite{Merton_jump1976} and more recently \cite{Linetsky2006}, among others.

The drawdown events are defined similarly as in \eqref{tauDk} where the log-price is now given by
\begin{displaymath}
\label{theproblem} \tilde{X}_t = \left \{
\begin{array}{ll}
\log \tilde{S}_0 +( r+\lambda-\frac{1}{2}\sigma^2) t + \sigma W_t,  & t < \zeta \\
-\infty, & t \ge \zeta,
\end{array}
\right.
\end{displaymath}
We follow a similar definition of the  drawdown insurance contract  from Section \ref{sect-formulation}. One major effect of a default event is that   it  causes drawdown  and the contract will expire immediately. In the perpetual case, the premium payment is paid until $\tdk\wedge \tuk$ if it happens before both the default time $\zeta$ and the maturity $T$, or until the default time $\zeta$ if $T\ge \tdk\wedge\tuk\ge\zeta$. Notice that, if   no drawup or drawdown of size $k$ happens before   $\zeta$,  then the  drawdown time $\tdk$ will coincide with the  default time, i.e. $\tdk=\zeta$.  The expected value to the buyer is given by
\begin{align}\label{1def}v (y, z ;\, p) = \E^{y,z}\left\{- \int_0^{\tdk \wedge \tuk \wedge \zeta\wedge T}e^{-rt}p dt + \alpha e^{-r\tdk} \1_{\{\tdk\le \tuk\wedge \zeta\wedge T\}}\right\}.
\end{align} Again,  the stopping times $\tau_D(k)$ and $\tau_U(k)$ based on $\tilde{X}$ do not depend on $x$, and therefore, the contract value $v$ is a function of the initial drawdown $y$  and drawup $z$.

Under this defaultable stock model, we obtain the following useful formula for the fair premium.

\begin{proposition} The fair premium for a drawdown insurance maturing at $T$, written on the defaultable stock in \eqref{s_def} is given by
  \begin{align}\label{pstar_deft}
    P^*=\frac{\alpha\{rL_{r+\lambda}^{T}+\lambda-\lambda R_{r+\lambda}^{T}-\lambda e^{-(r+\lambda)T}\Q^{y,z}(\tdk\wedge\tuk\ge T)\}}{1-L_{r+\lambda}^{T}-R_{r+\lambda}^{T}-e^{-(r+\lambda)T}\Q^{y,z}(\tdk\wedge\tuk\ge T)},\end{align}
  where $L_{r+\lambda}^{T}$ and $R_{r+\lambda}^T$ are given in \eqref{truncatedL} and \eqref{Rr}, respectively.

\end{proposition}

\begin{proof}
As seen in \eqref{pstar}, the fair premium $P^*$ satisfies
\begin{align}
  P^*=\frac{r\alpha \E^{y,z}\{e^{-r\tdk}\1_{\{\tdk\le \tuk \wedge \zeta \wedge T\}}\}}{1-\E^{y,z}\{e^{-r(\tdk\wedge \tuk\wedge\zeta\wedge T)}\}}.
\end{align}
We first compute the expectation in the numerator.
\begin{align}
  \notag &\E^{y,z}\{e^{-r\tdk}\1_{\{\tdk\le \tuk \wedge \zeta \wedge T\}}\}\\
  =&\int_0^T\lambda e^{-\lambda t}\E^{y,z}\{e^{-r\tdk}\1_{\{\tdk\le \tuk \wedge t\}}\}dt+\E^{y,z}\{e^{-r\tdk}\1_{\{\tdk \le \tuk\wedge T\}}\}\cdot\Q^{y,z}(\zeta>T)+\notag\\
&+\E^{y,z}\{e^{-r\zeta}\1_{\{\tdk\wedge \tuk\ge \zeta, \zeta<T\}}\}\notag \\
  =&\int_0^T e^{-(r+\lambda)s}\frac{\partial}{\partial s}\Q^{y,z}(\tdk\le \tuk\wedge s)ds+\int_0^T \lambda e^{-(r+\lambda)t}\Q^{y,z}(\tdk\wedge\tuk \ge t)dt\notag\\
  =&L_{r+\lambda}^{T}+\frac{\lambda}{r+\lambda}\bigg\{1-e^{-(r+\lambda)T}\Q^{y,z}(\tdk\wedge\tuk\ge T)+\int_0^T e^{-(r+\lambda)t}\frac{\partial}{\partial t}\Q^{y,z}(\tdk\wedge\tuk\ge t)dt\bigg\}\notag\\
  =&L_{r+\lambda}^{T}+\frac{\lambda}{r+\lambda}\{1-e^{-(r+\lambda)T}\Q^{y,z}(\tdk\wedge\tuk\ge T)-L_{r+\lambda}^{T}-R_{r+\lambda}^{T}\}.\label{pf1t}
\end{align}
Next,  the Laplace transform of $\tdk\wedge \tuk\wedge\zeta\wedge T$ is given by
\begin{align}
  \notag &\E^{y,z}\{e^{-r(\tdk\wedge \tuk\wedge\zeta\wedge T)}\}\\
  =&\E^{y,z}\{e^{-r(\tdk\wedge\tuk)}\1_{\{\tdk\wedge \tuk< \zeta\wedge T\}}\}+e^{-rT}\Q^{y,z}(\tdk\wedge\tuk>T,\zeta>T)\notag\\
  &+\E^{y,z}\{e^{-r\zeta}\1_{\{\tdk\wedge \tuk\ge \zeta, \zeta<T\}}\}\notag \\
  =&\int_0^T e^{-(r+\lambda)s}\frac{\partial}{\partial s}\Q^{y,z}(\tdk\wedge\tuk\le s)ds+e^{-(r+\lambda)T}\Q^{y,z}(\tdk\wedge\tuk>T)+\notag\\
&+\frac{\lambda}{r+\lambda}\{1-e^{-(r+\lambda)T}\Q^{y,z}(\tdk\wedge\tuk\ge T)-L_{r+\lambda}^{T}-R_{r+\lambda}^{T}\}\notag\\
=&L_{r+\lambda}^{T}+R_{r+\lambda}^{T}+\frac{\lambda}{r+\lambda}\{1-L_{r+\lambda}^{T}-R_{r+\lambda}^{T}\}+e^{-(r+\lambda)T}\frac{r}{r+\lambda}\Q^{y,z}(\tdk\wedge\tuk\ge T)\notag\\
  =&\frac{\lambda}{r+\lambda}+\frac{r}{r+\lambda}\{L_{r+\lambda}^{T}+R_{r+\lambda}^{T}+e^{-(r+\lambda)T}\Q^{y,z}(\tdk\wedge\tuk\ge T)\}.\label{pf2t}
  \end{align}
  Rearranging \eqref{pf1t} and \eqref{pf2t} yields \eqref{pstar_deft}.\end{proof}

By taking $T \to \infty$ in \eqref{pstar_deft}, we obtain  the fair premium for the perpetual drawdown insurance  in closed form.
 \begin{proposition} The fair premium for the perpetual drawdown insurance  written on the defaultable stock in \eqref{s_def} is given by
   \begin{align}\label{pstar_def}P^*=\frac{\alpha \left(rL_{r+\lambda}^\infty+\lambda-\lambda R_{r+\lambda}^\infty\right)}{1-L_{r+\lambda}^\infty-R_{r+\lambda}^\infty},\end{align} where $L_{r+\lambda}^\infty$ and $R_{r+\lambda}^\infty$ are given in \eqref{L}.
 \end{proposition}

In Figure \ref{fig_def}, we illustrate  for the perpetual case that the fair premium  is increasing with the default intensity $\lambda$ and  approaches $\alpha \lambda$ for high default risk. This observation, which can be formally shown by taking the limit in \eqref{pstar_def}, is intuitive since high default risk implies that a drawdown will more likely happen and that it is most likely triggered by a default.

\begin{figure}
  \begin{center}
    \includegraphics[width=3in]{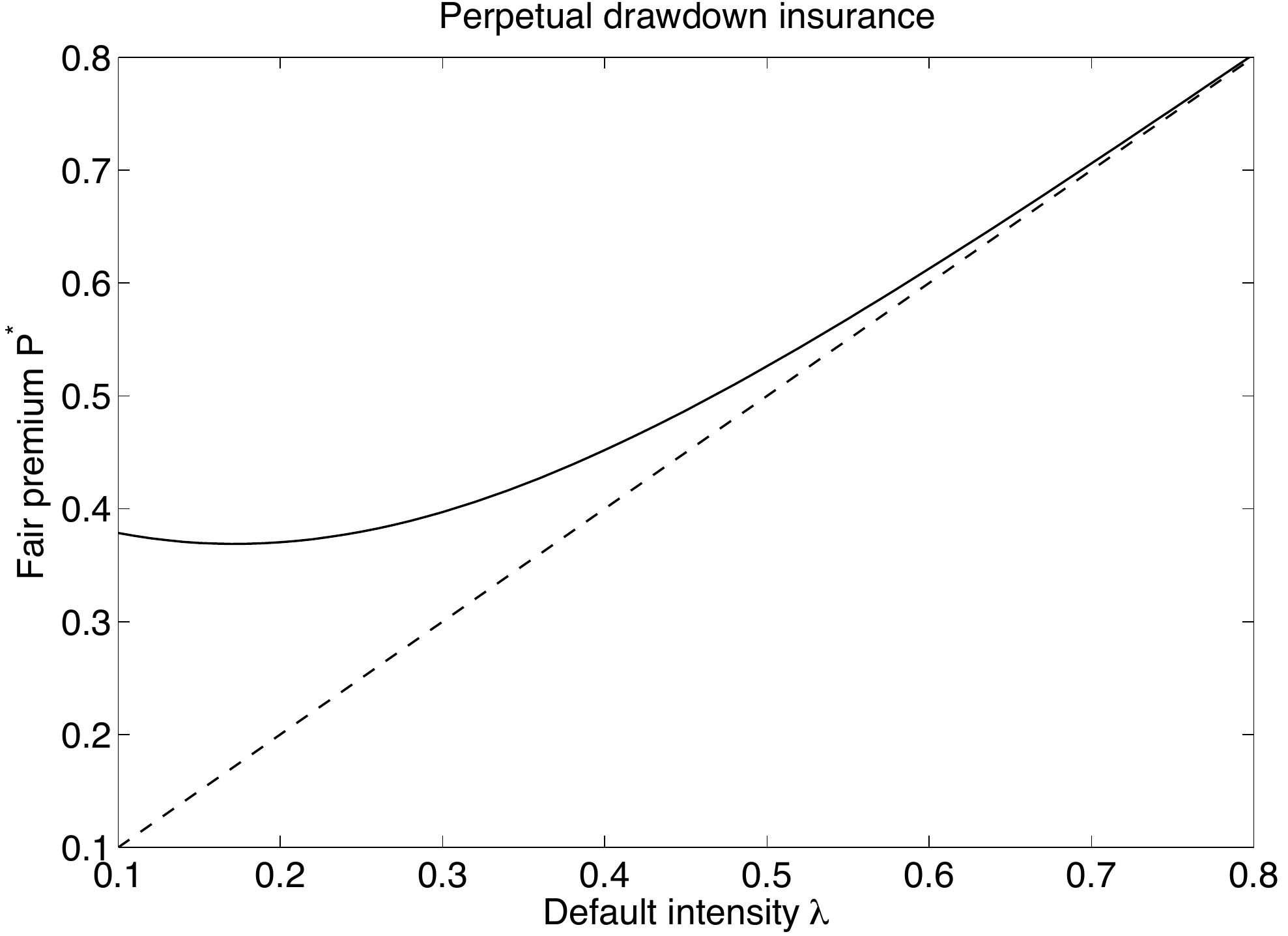}
    \begin{small}
    \caption{The fair premium (solid) as a function of the default intensity $\lambda$, which dominates the straight dashed line $\alpha \lambda$. As $\lambda\to\infty$, the fair premium $P^*\to\alpha\lambda$. Parameters: $r=2\%, \sigma=30\%, y=z=0.1,  k=0.5, \alpha=1$.   }\label{fig_def}\end{small}
  \end{center}
  \end{figure}

\section{Conclusions}\label{sect-conclude} We have studied the practicality of insuring against market crashes and proposed a number of tractable  ways to value drawdown protection. Under the geometric Brownian motion dynamics, we provided the formulas for the fair premium for  a number of insurance contracts, and examine its behavior with respect to key model parameters. In the cancellable drawdown insurance, we showed that the protection buyer would monitor the drawdown process and optimally stop  the premium payment as the drawdown  risk diminished.   Also, we investigated the impact of  default risk on  drawdown and derived  analytical formulas for  the fair premium.

For future research, we envision that the valuation and optimal stopping problems herein can be studied under other price dynamics, especially when  drawdown formulas, e.g.  for  Laplace transforms and hitting time distributions, are available (see \cite{PospVeceHadj} for the diffusion case).   Although we have  focused our analysis   on drawdown insurance written on a single underlying  asset, it is both interesting  and challenging  to model drawdowns across   multiple financial markets, and investigate the systemic impact of a drawdown occurred in one market. This would involve modeling the interactions among various financial markets \cite{Eisenberg} and developing new measures of systemic risk \cite{Markus}. Lastly, the idea of market drawdown and  the associated mathematical tools can also be useful in other areas, such as portfolio optimization problems \cite{GrosZhou,CvitKara}, risk management \cite{ChekUryaZaba}, and signal detection \cite{Poor2008}.\\

\textbf{Acknowledgement.} We are grateful to the seminar participants at Johns Hopkins University and
Columbia University. We also thank  INFORMS for the Best Presentation Award for this work at the 2011 Annual Meeting. Tim Leung's work is  supported by NSF grant DMS-0908295. Olympia Hadjiliadis' work is  supported by NSF grants CCF-MSC-0916452, DMS-1222526 and PSC-CUNY grant 65625-00 43. Finally, we thank   the Editor and anonymous  referees for  their useful remarks and suggestions.

\section{Proof of Lemmas}\label{sect_proofs}
  \subsection{Conditional Laplace Transform of Drawdown Time}\label{sect-condtranf}
  In order to prepare for our subsequent proofs on the cancellable drawdown insurance in Section \ref{sect-cancellableinsurance}, we now summarize a number properties of the conditional Laplace transform of $\tau_D(k)$ (see \eqref{def_xi}).

\begin{proposition}
The conditional Laplace transform function $\xi(\cdot)$ has the following properties:
\begin{enumerate}
\item $\xi(\cdot)$ is positive and  increasing on $(0,k)$.
\item $\xi(\cdot)$ satisfies differential equation
  \begin{eqnarray}\label{ode}\frac{1}{2}\sigma^2\xi''(y)-\mu\xi'(y)=r\xi(y),\end{eqnarray}
  with the Neumann condition
  \[\xi^{'}(0)=0.\]
\item $\xi(\cdot)$ is strictly convex, i.e., $\xi^{''}(y)>0$ for $y\in(0,k)$.
  \end{enumerate}
\end{proposition}
\begin{proof} Property (i) follows directly from the definition of $\xi(y)$ and strong Markov property. Property (ii) follows directly from differentiation of \eqref{zdprop}. For property (iii), the proof is as follows. If $\mu\ge 0$, then  (\ref{ode}) implies that
  \[\xi^{''}(y)=\frac{2\mu}{\sigma^2}\xi^{'}(y)+\frac{2r}{\sigma^2}\xi(y)>0, \quad y\in(0,k).\]
  If $\mu<0$, then  (\ref{zd})  and (\ref{ode}) imply that for $y\in(0,k)$,
  \begin{eqnarray}
\xi^{'}(y)&=&\bigg(\Xi_{\mu,\sigma}^r+\frac{\mu}{\sigma^2}\bigg)\left(\xi(y)-e^{(\frac{\mu}{\sigma^2}-\Xi_{\mu,\sigma}^r)y}\xi(0)\right),\\
    \xi^{''}(y)&=&\frac{2\mu}{\sigma^2}\xi^{'}(y)+\frac{2r}{\sigma^2}\xi(y)\notag\\
    &=&\bigg(\Xi_{\mu,\sigma}^r+\frac{\mu}{\sigma^2}\bigg)^2\xi(y)-\frac{2\mu}{\sigma^2}\bigg(\Xi_{\mu,\sigma}^r+\frac{\mu}{\sigma^2}\bigg)
    e^{(\frac{\mu}{\sigma^2}-\Xi_{\mu,\sigma}^r)y}\xi(0)\,>0.
    \end{eqnarray}The last inequality follows from the fact that $\mu<0$ and $\Xi_{\mu,\sigma}^r+ \frac{\mu}{\sigma^2}>0$.  Hence, strict convexity follows.
\end{proof}

\subsection{Proof of Lemma \ref{prop-unique}}\label{proof-prop-unique}
In view of \eqref{pasting}, we seek the root  $\theta^*$ of the equation:
\begin{align}\label{Ftheta}F(\theta)\mathop{:=}\bigg(\frac{\mu}{\sigma^2}-\Xi_{\mu,\sigma}^r\coth(\Xi_{\mu,\sigma}^r(k-\theta))\bigg)\tilde{f}(\theta)-\tilde{f}^{'}(\theta) =0 .\end{align}
 To this end, we compute
\begin{align}\label{Fprime}F^{'}(\theta)=-\frac{(\s)^2}{(\sinh(\s(k-\theta))^2}\tilde{f}(\theta)-\bigg(\Xi_{\mu,\sigma}^r\coth(\Xi_{\mu,\sigma}^r(k-\theta))
-\frac{\mu}{\sigma^2}\bigg)\tilde{f}^{'}(\theta)-\tilde{f}^{''}(\theta).\end{align}
Since $f$ is monotonically decreasing from $\tilde{f}(0)>0$ to $\tilde{f}(k)=-\alpha-c<0$, there exists a unique $\theta_0\in(0,k)$ such that $\tilde{f}(\theta_0)=0$. We have   $F(\theta_0)=-\tilde{f}^{'}(\theta_0)>0$ by \eqref{Ftheta} and $F(0)=(\frac{\mu}{\sigma}-\s\coth(\s k))\tilde{f}(0)<0$, which implies  that $F(\theta)=0$ has at least one solution $\theta^*\in(0,\theta_0)$.  Moreover, for  $\theta\in(\theta_0,k)$, $\tilde{f}(\theta)<0$ and hence $F(\theta)>0$ by \eqref{Ftheta}, there is no root in $(\theta_0,k)$.

Next,  we show the root is unique by proving that $F^{'}(\theta)>0$ for all $\theta\in (0,\theta_0)$.   To this end, we first observe from \eqref{F1} that $\tilde{f}$ can be expressed as  $\tilde{f}(\theta)=C(\xi(\theta_0)-\xi(\theta))$, for $\theta,\theta_0\in(0,k)$, where $C=(\alpha+\frac{p}{r})>0$ and $C\xi(\theta_0)=\frac{p}{r}-c$. Putting these into \eqref{Fprime}, we express $F^{'}(\theta)$ in terms $\xi$ instead of $\tilde{f}$. In turn, verifying $F^{'}(\theta)>0$ is reduced to:

\begin{lemma}\label{lem1}
  \[\inf_{0<\theta<\theta_0<k}
  \left\{\xi^{''}(\theta)+\left(\s\coth(\s(k-\theta))-\frac{\mu}{\sigma^2}\right)\xi^{'}(\theta)+\frac{(\s)^2}{(\sinh(\s(k-\theta)))^2}(\xi(\theta)-\xi(\theta_0))\right\}\ge0,\]
and the infimum is attained at $\theta = \theta_0=k$.\end{lemma}
\begin{proof}
We begin by using (\ref{ode}) to rewrite the statement in the lemma as
\begin{multline*}\inf_{0<\theta<\theta_0<k}\bigg\{\left(\s\coth(\s(k-\theta))+\frac{\mu}{\sigma^2}\right)\xi^{'}(\theta)+\left((\s)^2\coth^2(\s(k-\theta))-\frac{\mu^2}{\sigma^4}\right)\xi(\theta)\\-\frac{(\s)^2}{\sinh^2(\s(k-\theta))}\xi(\theta_0)\bigg\}\ge0\end{multline*}
 By the strong Markov property of process $D_\cdot$, the function $\xi$ satisfies a more general version of \eqref{zd}. Specifically, for $0\le y_1, y_2<k$,
\begin{eqnarray}\label{xy}
\xi(y_2)&=&e^{\frac{\mu}{\sigma^2}(y_2-k)}\frac{\sinh(\s(y_2-y_1))}{\sinh(\s(k-y_1))}+e^{\frac{\mu}{\sigma^2}(y_2-y_1)}\frac{\sinh(\s(k-y_2))}{\sinh(\s(k-y_1))}\xi(y_1).
\end{eqnarray}
Define for $y\in[0,k)$,
\begin{eqnarray}\label{lambda}
  \Lambda(y)&=&\frac{e^{-\frac{\mu y }{\sigma^2}}\xi(y)}{\sinh(\s(k-y))}.
\end{eqnarray}
Then function $\Lambda(\cdot)$ satisfies (see (\ref{xy}))
\begin{eqnarray}
  \Lambda(y_2)-\Lambda(y_1)&=&\frac{e^{-\frac{\mu k}{\sigma^2}}\sinh(\s(y_2-y_1))}{\sinh(\s(k-y_1))\cdot\sinh(\s(k-y_2))},\quad \forall y_1,y_2\in[0,k),
\end{eqnarray}
from which we can easily obtain that
\begin{eqnarray}
\Lambda^{'}(y)&=&\frac{\s  e^{-\frac{\mu k}{\sigma^2}}}{\sinh^2(\s(k-y))}>0,\quad \forall y\in[0,k).
  \end{eqnarray}
Straightforward computation shows that
  \[\Lambda^{'}(y)=e^{-\frac{\mu y}{\sigma^2}}\frac{(\s\coth(\s(k-y))-\frac{\mu}{\sigma^2})\xi(y)+\xi^{'}(y)}{\sinh(\s(k-y))}>0, \quad \forall y\in[0,k).\]
  Thus,
  \begin{eqnarray}\label{zeta2lambda}\xi^{'}(y)=\Lambda^{'}(y)e^{\frac{\mu y}{\sigma^2}}\sinh(\s(k-y))-(\s\coth(\s(k-y))-\frac{\mu}{\sigma^2})\xi(y).\end{eqnarray}

Using (\ref{zeta2lambda}), the above inequality is equivalent to
\[\inf_{0<\theta<\theta_0<k}\bigg\{\Lambda^{'}(\theta)\bigg(e^{\frac{\mu\theta}{\sigma^2}}\left(\s\cosh(\s(k-\theta))+\frac{\mu}{\sigma^2}\sinh(\s(k-\theta))\right)-\s e^{\frac{\mu k}{\sigma^2}}\xi(\theta_0)\bigg)\bigg\}\ge 0.\]
Let us denote by
\[H(\theta,\theta_0)=e^{\frac{\mu\theta}{\sigma^2}}\left(\s\cosh(\s(k-\theta))+\frac{\mu}{\sigma^2}\sinh(\s(k-\theta))\right)-\s e^{\frac{\mu k}{\sigma^2}}\xi(\theta_0).\]
We will show that \[\inf_{0<\theta<\theta_0<k}H(\theta,\theta_0)\ge 0.\]
Notice that for $\theta\in[0,\theta_0]$,
\[\frac{\partial H}{\partial\theta}=-\frac{2r}{\sigma^2}e^{\frac{\mu\theta}{\sigma^2}}\sinh(\s(k-\theta))<0,\]
therefore
\[\inf_{0\le \theta\le\theta_0}H(\theta,\theta_0)=H(\theta_0,\theta_0).\]
Moreover,
\[\frac{\partial}{\partial\theta_0}H(\theta_0,\theta_0)=-\frac{2r}{\sigma^2}e^{\frac{\mu\theta_0}{\sigma^2}}\sinh(\s(k-\theta_0))-\s e^{\frac{\mu k}{\sigma^2}}\xi^{'}(\theta_0)<0.\]
As a result,
\[\inf_{0\le\theta\le\theta_0<k}H(\theta,\theta_0)=H(k,k)=0.\]
This completes the proof of Lemma \ref{lem1}.
\end{proof}
 Since our problem concerns $\theta < \theta_0<k$,   Lemma \ref{lem1} says $F^{'}(\theta)>0$ for $\theta \in (0, \theta_0)$, which confirms that there is at most one solution to equation $F(\theta)=0$. This concludes the uniqueness of smooth pasting point $\theta^*\in(0,\theta_0)$.

\subsection{Proof of Lemma \ref{lem2}}\label{sect-appx-gf}
\begin{proof}
Let us consider
\[J(y):=g(y;\theta^*)-\tilde{f}(y)=C\left(\beta(y)(\xi(\theta_0)-\xi(\theta^*))+\xi(y)-\xi(\theta_0)\right),~y\in[\theta^*,k).\]
We check its derivatives with respect to $x$:
\begin{eqnarray}
  J^{'}(y)&=&C\left(\beta^{'}(y)(\xi(\theta_0)-\xi(\theta^*))+\xi^{'}(y)\right),\\
  J^{''}(y)&=&C\left(\beta^{''}(y)(\xi(\theta_0)-\xi(\theta^*))+\xi^{''}(y)\right),
  \end{eqnarray}where
  \begin{eqnarray}\label{xi}
\beta(y)&=&\frac{g(y;\theta^*)}{f(\theta^*)}=e^{\frac{\mu}{\sigma^2}(y-\theta^*)}\frac{\sinh(\s(k-y))}{\sinh(\s(k-\theta^*))}, \quad  y\in(\theta^*,k).
  \end{eqnarray}

  Using probabilistic nature of function $\beta(\cdot)$ we know that it is positive and decreasing. Therefore, if $\mu\le 0$, we have
  \[\beta^{''}(y)=\frac{2\mu}{\sigma^2}\beta^{'}(y)+\frac{2r}{\sigma^2}\beta(y)>0~\Rightarrow J^{''}(y)\ge C\xi^{''}(y)>0.\]
    On the other hand, if $\mu>0$, from (\ref{xi}) we have
    \begin{eqnarray*}
      \beta^{'}(y)&=&\left(\frac{\mu}{\sigma^2}-\s\right)\beta(y)+\frac{\s e^{\frac{\mu}{\sigma^2}(y-\theta^*)-\s(k-y)}}{\sinh(\s(k-\theta^*))},\\
      \beta^{''}(y)&=&\frac{2\mu}{\sigma^2}\beta^{'}(y)+\frac{2r}{\sigma^2}\beta(y)=\left(\s-\frac{\mu}{\sigma^2}\right)^2\beta(y)+\frac{2\mu}{\sigma^2}\frac{\s e^{\frac{\mu}{\sigma^2}(y-\theta^*)-\s(k-y)}}{\sinh(\s(k-\theta^*))}>0,\\
      &\Rightarrow& J^{''}(y)\ge C\xi^{''}(y)>0.
      \end{eqnarray*}
      So in either case ($\mu\le 0$ or $\mu>0$), $J^{'}(\cdot)$ is an increasing function, and
      \[J^{'}(y)>J^{'}(\theta^*)=0,\quad\forall y\in(\theta^*,k),\]
which implies that
\[J(y)>J(\theta^*)=0,\quad \forall y\in(\theta^*,k).\]
This completes the proof.
\end{proof}

\end{document}